\documentclass[a4paper]{llncs}
\usepackage[utf8]{inputenc}
\title{The supersingular isogeny problem in genus~2 and beyond}

\author{Craig Costello\inst{1} \and Benjamin Smith\inst{2}}

\institute{ 
    Microsoft Research, USA \\
    \email{craigco@microsoft.com} \\
    \and
    Inria 
    \emph{and}
    École polytechnique,
    Institut Polytechnique de Paris,
    Palaiseau, France \\
    \email{smith@lix.polytechnique.fr}  
}

\date{\today}

\usepackage{amsmath,amsfonts,amssymb}
\usepackage{url}
\usepackage{hyperref}
\usepackage{multirow}
\usepackage{doi}
\usepackage{xcolor}
\usepackage{booktabs}
\usepackage{xspace}
\usepackage[all,cmtip]{xy}
\usepackage[ruled,vlined,linesnumbered]{algorithm2e}
\usepackage{fancyvrb}
\DontPrintSemicolon

\newcommand{\FF}{\ensuremath{\mathbb{F}}}
\newcommand{\FFbar}{\ensuremath{\overline{\mathbb{F}}}}

\newcommand{\ZZ}{\ensuremath{\mathbb{Z}}}
\newcommand{\EC}{\ensuremath{\mathcal{E}}}
\newcommand{\XC}{\ensuremath{\mathcal{C}}}

\newcommand{\AV}{\ensuremath{\mathcal{A}}}
\newcommand{\BV}{\ensuremath{\mathcal{B}}}

\newcommand{\Jac}[1]{\ensuremath{\mathcal{J}_{#1}}}
\newcommand{\Aut}{\ensuremath{\mathrm{Aut}}}
\newcommand{\End}{\ensuremath{\mathrm{End}}}
\newcommand{\Pic}{\ensuremath{\mathrm{Pic}}}

\newcommand{\dualof}[1]{\ensuremath{{#1}^\dagger}}

\newcommand{\subgrp}[1]{\ensuremath{\langle{#1}\rangle}}
\newcommand{\qbinom}[3][\ell]{\ensuremath{{\genfrac{[}{]}{0pt}{}{{#2}}{{#3}}}_{#1}}}

\newcommand{\softO}{\ensuremath{\widetilde{O}}}

\SetKwProg{Fn}{function}{}{}
\SetKwFunction{KeyPair}{KeyPair}
\SetKwFunction{DH}{DH}
\SetKwFunction{BSGS}{BSGS}
\SetKwFunction{PohligHellman}{PohligHellman}

\newcommand{\SSset}[1]{\ensuremath{S_{#1}(p)}}
\newcommand{\SSgraph}[2]{\ensuremath{\Gamma_{#1}({#2};p)}}

\spnewtheorem{hypothesis}{Hypothesis}{\bfseries}{\itshape}

\begin{document}
\maketitle

\begin{abstract}
    Let $A/\overline{\mathbb{F}}_p$ and $A'/\overline{\mathbb{F}}_p$
    be superspecial principally polarized abelian varieties
    of dimension $g>1$.
    For any prime $\ell \ne p$, we give an algorithm that finds
    a path $\phi \colon A \rightarrow A'$ 
    in the $(\ell, \dots , \ell)$-isogeny graph 
    in $\widetilde{O}(p^{g-1})$ group operations on a classical computer, and
    $\widetilde{O}(\sqrt{p^{g-1}})$ calls to the Grover oracle on a quantum
    computer. The idea is to find paths from $A$ and $A'$ 
    to nodes that correspond to products of lower dimensional
    abelian varieties, and to recurse down in dimension until an
    elliptic path-finding algorithm (such as Delfs--Galbraith) 
    can be invoked to connect the paths in dimension $g=1$. 
    In the general case where $A$ and $A'$ are any two nodes in the graph,
    this algorithm presents an asymptotic improvement over all of the
    algorithms in the current literature. In the special case where $A$
    and $A'$ are a known and relatively small number of steps away from
    each other (as is the case in higher dimensional analogues of SIDH),
    it gives an asymptotic improvement over the quantum claw finding
    algorithms and an asymptotic improvement over the classical van
    Oorschot--Wiener algorithm. 
%    Our results provide evidence that performance trade-offs for
%    high-dimensional isogeny-based cryptography are unlikely to be
%    favourable when venturing beyond dimension $g=2$. 
\end{abstract}

\section{%%%%%%%%%%%%%%%%%%%%%%%%%%%%%%%%%%%%%%%%%%%%%%%%%%%%%%%%%%%%%%%%%%%%%%%
    Introduction
}%%%%%%%%%%%%%%%%%%%%%%%%%%%%%%%%%%%%%%%%%%%%%%%%%%%%%%%%%%%%%%%%%%%%%%%%%%%%%%%

Isogenies of supersingular elliptic curves
are now well-established in %number-theoretic
cryptography,
from the Charles--Goren--Lauter hash function~\cite{CGL}
to Jao and De Feo's SIDH key exchange~\cite{SIDH}
and beyond~\cite{SIKE,GPS17,seasign,vdf}.
While the security of isogeny-based cryptosystems depend 
on the difficulty of a range of computational problems,
the fundamental one is the \emph{isogeny problem}:
given supersingular elliptic curves \(\EC_1\) and \(\EC_2\) over
\(\FF_{p^2}\),
find a walk in the \(\ell\)-isogeny graph
connecting them.

One intriguing aspect of isogeny-based cryptography
is the transfer of elliptic-curve techniques
from classic discrete-log-based cryptography
into the post-quantum arena.
In this spirit, it is natural to consider cryptosystems
based on isogeny graphs of higher-dimensional abelian varieties,
mirroring the transition from elliptic (ECC)
to hyperelliptic-curve cryptography (HECC).
Compared with elliptic supersingular isogeny graphs,
the higher-dimensional graphs have more vertices and higher degrees
for a given \(p\), which allows some interesting tradeoffs
(for example: in dimension \(g = 2\), we get the same number of vertices
with a \(p\) of one-third the bitlength).

For \(g = 2\),
Takashima~\cite{Takashima}
and Castryck, Decru, and Smith~\cite{Castryck--Decru--Smith}
have defined CGL-style hash functions,
while Costello~\cite{costellokummer} and Flynn and Ti~\cite{flynnti}
have already proposed SIDH-like key exchanges.
Generalizations to dimensions \(g > 2\),
using isogeny algorithms such as those in~\cite{AVIsogenies},
are easy to anticipate;
for example,
a family of hash functions on isogeny graphs of superspecial abelian
varieties with real multiplication was hinted at in~\cite{CGL-2}.

So far, when estimating security levels, these generalizations assume
that the higher-dimensional supersingular isogeny problem 
is basically as hard as the elliptic supersingular isogeny problem
in graphs of the same size.
%The one new attack in \(g=2\) is that we can find trivial cycles
%of length \(4\) in the graph, 
%so simply avoiding backtracking is not enough to avoid hash collisions.
%Various workarounds to avoid this are proposed in~\cite{flynnti}
%and~\cite{Castryck--Decru--Smith}; ultimately,
%it represents more of a minor inconvenience than a fatal flaw.
%
In this article,
we show that this assumption is false.
The general supersingular isogeny problem can be partially reduced to
a series of lower-dimensional isogeny problems,
and thus recursively to a series of elliptic isogeny problems.
%Our main result is the following:

\begin{theorem}\label{thm:classic}
    There exists a classical algorithm which,
    given a prime \(\ell\) and superspecial abelian varieties \(\AV_1\) and \(\AV_2\)
    of dimension \(g\) over \(\FFbar_p\) with \(p \not= \ell\),
    succeeds with probability \(\ge 1/2^{g-1}\) in computing
    a composition of \((\ell,\ldots,\ell)\)-isogenies
    from \(\AV_1\) to \(\AV_2\),
    running in expected time
    \(\softO((p^{g-1}/P))\) on \(P\) processors
    as \(p \to \infty\) (with \(\ell\) fixed).
\end{theorem}

Given that these graphs have \(O(p^{g(g+1)/2})\) vertices,
the expected runtime for generic random-walk algorithms
is \(\softO(p^{g(g+1)/4}/P)\).
Our algorithm therefore represents a substantial speedup,
with nontrivial consequences for cryptographic parameter
selection.\footnote{%
    Our algorithms apply to the full superspecial graph;
    we do not claim any impact on cryptosystems that run in 
    small and special subgraphs, such as CSIDH~\cite{csidh}.
}
We also see an improvement in quantum algorithms:

\begin{theorem}\label{thm:quantum}
    There exists a quantum algorithm which,
    given a prime \(\ell\) and superspecial abelian varieties \(\AV_1\) and \(\AV_2\)
    of dimension \(g\) over \(\FFbar_p\) with \(p \not= \ell\),
    computes a composition of \((\ell,\ldots,\ell)\)-isogenies from \(\AV_1\) to \(\AV_2\) 
    running in expected time \(\softO(\sqrt{p^{g-1}})\)
    as \(p \to \infty\) (with \(\ell\) fixed).
\end{theorem}

This reflects the general pattern seen in the passage from ECC to HECC:
the dimension grows, the base field shrinks----and the mathematical
structures become more complicated, which can ultimately reduce claimed security levels.
Just as index calculus attacks on
discrete logarithms become more powerful in higher genus, where useful
structures appear in Jacobians~\cite{Diem06,GTTD,Gaudry09,Smith09},
so interesting structures in higher-dimensional isogeny graphs
provide attacks that become more powerful as the dimension grows.
Here, the interesting structures are (relatively large) subgraphs corresponding to 
increasing numbers of elliptic factors in (polarized) abelian varieties.
These subgraphs are relatively large,
and so random-walking into them is relatively easy.
We can then glue together elliptic isogenies, found with an elliptic
path-finding algorithm, to form product isogenies between products
of elliptic curves, and thus to solve the original isogeny problem.
We will see that the path-finding problem in the superspecial graph gets
asymptotically easier as the dimension grows.

\paragraph{Notation and conventions.}
Throughout, \(p\) denotes a prime \(>3\),
and \(\ell\) a prime not equal to \(p\).
Typically, \(p\) is large, 
and \(\ell \ll \log(p)\) is small enough that 
computing \((\ell,\ldots,\ell)\)-isogenies
of \(g\)-dimensional principally polarized abelian varieties (PPAVs)
is polynomial in \(\log(p)\).
Similarly, 
we work with PPAVs in dimensions \(g \ll \log p\);
in our asymptotics and complexities, \(g\) and \(\ell\) are fixed.
We say a function \(f(X)\) is in \(\softO(g(X))\)
if \(f(X) = O(h(\log X)g(X))\) for some polynomial~\(h\).

\section{%%%%%%%%%%%%%%%%%%%%%%%%%%%%%%%%%%%%%%%%%%%%%%%%%%%%%%%%%%%%%%%%%%%%%%%
    The elliptic supersingular isogeny graph
}%%%%%%%%%%%%%%%%%%%%%%%%%%%%%%%%%%%%%%%%%%%%%%%%%%%%%%%%%%%%%%%%%%%%%%%%%%%%%%%

An elliptic curve \(\EC/\FFbar_p\)
is \emph{supersingular}
if \(\EC[p](\FFbar_p) = 0\).
We have a number of efficient algorithms
for testing supersingularity:
see Sutherland~\cite{Sutherland} for discussion.

Supersingularity is isomorphism-invariant,
and any supersingular \(\EC\) has \(j\)-invariant \(j(\EC)\) in \(\FF_{p^2}\);
and in fact the curve \(\EC\) 
can be defined over \(\FF_{p^2}\).
We let 
\[
    \SSset{1} 
    :=
    \left\{
        j(\EC) : \EC/\FF_{p^2} \text{ is supersingular}
    \right\}
    \subset
    \FF_{p^2}
\]
be the set of isomorphism classes of supersingular elliptic curves over \(\FFbar_p\).
It is well-known that
\begin{equation}
    \label{eq:size-of-S1}
    \#\SSset{1} 
    =
    \left\lfloor\frac{p}{12}\right\rfloor
    +
    \epsilon_p
\end{equation}
where \(\epsilon_p = 0\) if \(p \equiv 1 \pmod{12}\),
\(2\) if \(p \equiv -1 \pmod{12}\),
and \(1\) otherwise.
%% RESTORE IN LONG VERSION:
%In terms of \(\log p\),
%the set \(\SSset{1}\) is exponentially large;
%but it is also an exponentially small subset of \(\FF_{p^2}\)
%(the set of \emph{all} \(j\)-invariants),
%in the sense that a uniformly random element of \(\FF_{p^2}\)
%is in \(\SSset{1}\) with probability only \(\sim 12/p\).

Now fix a prime \(\ell \not= p\),
and consider the directed multigraph
\(\SSgraph{1}{\ell}\)
whose vertex set is \(\SSset{1}\),
and whose edges correspond to \(\ell\)-isogenies
between curves (again, up to isomorphism).
The graph 
\(\SSgraph{1}{\ell}\) is \((\ell+1)\)-regular:
there are (up to isomorphism)
\(\ell+1\) distinct \(\ell\)-isogenies
from a supersingular elliptic curve \(\EC/\FF_{p^2}\)
to other elliptic curves,
corresponding to the \(\ell+1\) order-\(\ell\) subgroups
of \(\EC[\ell](\FFbar_p) \cong (\ZZ/\ell\ZZ)^2\)
that form their kernels.
But since supersingularity is isogeny-invariant,
the codomain of each isogeny is again supersingular;
that is,
the \(\ell+1\) order-\(\ell\) subgroups of \(\EC[\ell]\)
are in bijection with
the edges out of \(j(\EC)\) in \(\SSgraph{1}{\ell}\).

\begin{definition}
    A \emph{walk} of length \(n\) in \(\SSgraph{1}{\ell}\)
    is a sequence of edges
    \(j_0 \to j_1 \to \cdots \to j_n\).
    A \emph{path} 
    in \(\SSgraph{1}{\ell}\)
    is an acyclic (and, in particular, non-backtracking) walk:
    that is, a walk 
    \(j_0 \to j_1 \to \cdots \to j_n\)
    such that \(j_i = j_{i'}\) if and only if \(i = i'\).
\end{definition}

Pizer~\cite{pizer}
proved that \(\SSgraph{1}{\ell}\) is Ramanujan:
in particular, 
\(\SSgraph{1}{\ell}\) is a connected expander graph,
% and \(\lambda(\SSgraph{1}{\ell}) = 2\sqrt{\ell}\).
%% FIXME: should that be \lambda(G)/d and not \lambda(G) ?
and its diameter %(the maximal distance between any two vertices)
is \(O(\log p)\).
We therefore expect the end-points of short random walks
from any given vertex \(j_0\) to
quickly yield a uniform distribution on \(\SSset{1}\).
Indeed, 
if \(j_0\) is fixed
and \(j_n\) is the end-point of an \(n\)-step
random walk from \(j_0\) in \(\SSgraph{1}{\ell}\),
then~\cite[Theorem 1]{GPS17}
shows that
\begin{equation}
    \label{eq:distribution-1}
    \left|
        \mathrm{Pr}[j_n = j] - \frac{1}{\#\SSset{1}}
    \right|
    \le
    \left(
        \frac{2\sqrt{\ell}}{\ell+1}
    \right)^n
    \qquad
    \text{for all }
    j\in \SSset{1}
    \,.
\end{equation}

The \emph{isogeny problem} in \(\SSgraph{1}{\ell}\) is,
given \(j_0\) and \(j\) in \(\SSset{1}\),
to find a path 
(of any length) 
from \(j_0\) to \(j\) 
in \(\SSgraph{1}{\ell}\).
The difficulty of the isogeny problem underpins
the security of the Charles--Goren--Lauter hash function
(see~\S\ref{sec:CGL} below).

The isogeny problem is supposed to be hard.
Our best generic classical path-finding algorithms
look for collisions in random walks,
and run in expected time 
the square root of the graph size:
in this case,
\(\softO(\sqrt{p})\).
In the special case of supersingular isogeny graphs,
we can make some practical improvements
but the asymptotic complexity remains the same:
given \(j_0\) and \(j\) in \(F_1(p;\ell)\),
we can compute a path \(j_0 \to j\)
in \(\softO(\sqrt{p})\) classical operations
(see~\cite{Delfs--Galbraith}).

The best known quantum algorithm for path-finding~\cite{BJS14} instead
searches for paths from \(j_0 \to j_0'\) and from \(j \to j'\), where
\(j_0'\) and \(j'\) are both in \(\FF_p\). Of the \(O(p)\) elements in
\(\SSset{1}\), there are \(O(\sqrt{p})\) elements contained in
\(\FF_p\); while a classical search for elements this sparse would
therefore run in time \(O(\sqrt{p})\), Grover's quantum
algorithm~\cite{grover} completes the search in expected time
\(O(\sqrt[4]{p})\). It remains to find a path from \(j_0'\) to \(j'\).
This could be computed classically in time \(\softO(\sqrt[4]{p})\) using
the Delfs--Galbraith algorithm, but Biasse, Jao and Sankar~\cite{BJS14}
show that a quantum computer can find paths between subfield curves in
subexponential time, yielding an overall algorithm that runs in expected
time \(O(\sqrt[4]{p})\).

We can also consider the problem of finding paths
of a fixed (and typically \emph{short}) length:
for example,
given \(e > 0\) and \(j_0\) and \(j\) in \(\SSset{1}\)
such that there exists a path \(\phi: j_0 \to \cdots \to j\)
of length \(e\),
find \(\phi\).
This problem arises in the security analysis of SIDH, for example.

\section{%%%%%%%%%%%%%%%%%%%%%%%%%%%%%%%%%%%%%%%%%%%%%%%%%%%%%%%%%%%%%%%%%%%%%%%
    Cryptosystems in the elliptic supersingular graph
}%%%%%%%%%%%%%%%%%%%%%%%%%%%%%%%%%%%%%%%%%%%%%%%%%%%%%%%%%%%%%%%%%%%%%%%%%%%%%%%
\label{sec:crypto-g-1}

\paragraph{The Charles--Goren--Lauter hash function (CGL).}
\label{sec:CGL}
Supersingular isogenies appeared in cryptography with the CGL hash function,
which operates in \(\SSgraph{1}{2}\).
Fix a base point \(j_0\) in \(\SSset{1}\),
and one of the three edges in \(\SSgraph{1}{2}\) leading 
into it: \(j_{-1} \to j_0\), say.
To hash an \(n\)-bit message \(m = (m_0,m_1,\ldots,m_{n-1})\),
we let \(m\) drive a non-backtracking walk
\(j_0 \to \cdots \to j_{n}\) on \(\SSgraph{1}{2}\):
for each \(0 \le i < n\),
we compute the two roots \(\alpha_0\) and \(\alpha_1\) of
\(\Phi_2(j_i,X)/(j_{i-1}-X)\)
to determine the neighbours of \(j_i\) that are not \(j_{i-1}\),
numbering the roots with respect to some ordering of~\(\FF_{p^2}\)
(here \(\Phi_2(Y,X)\) is the classical modular polynomial),
and set \(j_{i+1} = \alpha_{m_i}\).

Once we have 
%used the entire input \(m = (m_0,\ldots,m_{n-1})\)
%to compute a walk
computed the entire walk
\(j_0 \to \cdots \to j_{n}\),
we can derive a \(\log_2p\)-bit hash value \(H(m)\)
from the end-point~\(j_{n}\);
we call this step \emph{finalisation}.
Charles, Goren, and Lauter suggest 
applying a linear function \(f: \FF_{p^2} \to \FF_p\)
to map \(j_{n}\) to \(H(m) = f(j_n)\).
For example, if \(\FF_{p^2} = \FF_p(\omega)\)
then we can map \(j_{n} = j_{n,0} + j_{n,1}\omega\)
(with \(j_{n,0}\) and \(j_{n,1}\) in \(\FF_p\))
to \(H(m) = aj_{n,0} + bj_{n,1}\)
for some fixed random choice
of \(a\) and \(b\) in \(\FF_p\).
Heuristically,
for general \(f\),
if we suppose \(\SSset{1}\) is distributed uniformly in~\(\FF_{p^2}\),
then roughly one in twelve elements of~\(\FF_p\)
appear as hash values,
and each of those has only one expected preimage in~\(\SSset{1}\).
%then every element of \(\FF_p\) appears as a hash value,
%and we expect roughly \(12\) walk end-points \(j_n\) in \(\FF_{p^2}\)
%to map to any given~\(h\) in \(\FF_p\).

Finding a preimage for a given hash value \(h\) in \(\FF_p\)
amounts to finding a path \(j_0 \to \cdots \to j\)
such that \(f(j) = h\):
that is, solving the isogeny problem.
%Recovering the message \((m_0,\ldots,m_{n-1})\) 
%from the path is just a simple matter of recoding.
We note that inverting the finalisation seems hard:
for linear \(f: \FF_p^2 \to \FF_p\),
we know of no efficient method
which given \(h\) in \(\FF_p\)
computes a supersingular \(j\) such that \(f(j) = h\).
(Brute force search requires \(O(p)\) trials.)
Finalisation thus gives us some protection against meet-in-the-middle
isogeny algorithms.
Finding collisions and second preimages for \(H\)
amounts to finding cycles in \(\SSgraph{1}{2}\).
For well-chosen \(p\) and \(j_0\),
this is roughly as hard as the isogeny problem~\cite[\S5]{CGL}.
%
%Finding a collision for \(H\)
%generally amounts to finding a cycle in \(\SSgraph{1}{\ell}\) containing~\(j_0\).
%Given a cycle \(j_0 \to j_1 \to \cdots \to j_n = j_0\),
%for each \(0 < k < n\)
%we have a collision between the message corresponding to
%the path \(j_0 \to \cdots \to j_k\),
%and the message corresponding to the \emph{dual} of the remainder of the
%cycle---that is, the path \(j_0 \to j_{n-1} \to \cdots \to j_k\).
%The collision can easily be extended by adding paths into \(j_0\)
%and out of \(j_k\) to both inputs.
%Finding a collision in the hash function yields a cycle in exactly the
%same way, as does finding second preimages.
%
%\Bnote{Relate collisions/preimages---see CGL}

%% RESTORE IN LONG VERSION:
%\begin{remark}
%    While many treatments of the CGL hash
%    omit the finalisation step,
%    it is an important part of the hash function:
%    it ensures that we output uniform-looking \(\log_2p\)-bit strings
%    (elements of \(\FF_{p}\))
%    rather than easily-distinguishable elements of \(\FF_{p^2}\),
%    which encode to \(2\log_2p\) bits.
%    Equation~\ref{eq:distribution-1}
%    implies that
%    the end-point \(j_n\) of the walk corresponding to
%    a \(O(\log_2p)\)-bit message is uniformly distributed
%    over a \(\sim p/12\)-element set:
%    that is, we have roughly \(\log_2p\) bits of entropy 
%    encoded as a \(2\log_2p\)-bit element of \(\FF_{p^2}\).
%    Finalisation removes this redundancy.
%\end{remark}

\paragraph{SIDH.}
Jao and De Feo's 
SIDH key exchange~\cite{SIDH}
begins with a supersingular curve \(\EC_0/\FF_{p^2}\),
where \(p\) is in the form \(c\cdot2^a3^b-1\),
with fixed torsion bases \(\subgrp{P_2,Q_2} = \EC_0[2^a]\)
and \(\subgrp{P_3,Q_3} = \EC_0[3^b]\)
(which are rational because of the special form of \(p\)).
Alice computes a secret walk \(\phi_A: \EC_0 \to \cdots \to \EC_A\) of length \(a\)
in \(\SSgraph{1}{2}\),
publishing~\(\EC_A\), \(\phi_A(P_3)\), and \(\phi_A(Q_3)\);
similarly,
Bob computes a secret walk \(\phi_B: \EC_0 \to \cdots \to \EC_B\) of length \(b\)
in \(\SSgraph{1}{3}\),
publishing~\(\EC_B\), \(\phi_B(P_2)\), and \(\phi_B(Q_2)\).
The basis images allow Alice to compute \(\phi_B(\ker\phi_A)\),
and Bob \(\phi_A(\ker\phi_B)\);
Alice can thus ``repeat'' her walk starting from \(\EC_B\),
and Bob his walk from \(\EC_A\),
to arrive at curves representing the same point in \(\SSset{1}\),
which is their shared secret.

Breaking Alice's public key amounts to solving an isogeny problem in
\(\SSgraph{1}{2}\) subject to the constraint that the walk have
length \(a\) (which is particularly short).
The \(3^b\)-torsion basis may give some useful
information here, though so far this is only exploited 
in attacks on artificial variants of SIDH~\cite{Petit17}.
Similarly, breaking Bob's public key amounts to solving a length-\(b\)
isogeny problem in \(\SSgraph{1}{3}\).
Alternatively, 
we can compute these short paths
by computing endomorphism rings:
\cite[Theorem 4.1]{GPST} states that
if \(\EC\) and \(\EC'\) are in \(\SSset{1}\)
and we have \emph{explicit descriptions}
of \(\End(\EC)\) and \(\End(\EC')\),
then we can efficiently compute the \emph{shortest} path
from \(\EC\) to \(\EC'\) in \(\SSgraph{1}{\ell}\)
(see~\cite{KLPT,GPST,petit} for further details on this approach).

%% RESTORE IN LONG VERSION:
%The explicit descriptions of \(\End(\EC)\) and \(\End(\EC')\)
%let us embed them as maximal orders in the same quaternion algebra;
%we can compute the quaternion ideal connecting them
%of smallest \(\ell\)-power-norm ideal
%in polynomial time;
%and again using the explicit endomorphism rings,
%we can translate this ideal
%effectively into an \(\ell\)-power isogeny.

\section{%%%%%%%%%%%%%%%%%%%%%%%%%%%%%%%%%%%%%%%%%%%%%%%%%%%%%%%%%%%%%%%%%%%%%%%
    Abelian varieties and polarizations
}%%%%%%%%%%%%%%%%%%%%%%%%%%%%%%%%%%%%%%%%%%%%%%%%%%%%%%%%%%%%%%%%%%%%%%%%%%%%%%%

An abelian variety 
is a smooth projective algebraic group variety.
An isogeny of abelian varieties
is a surjective finite morphism \(\phi: \AV \to \AV'\) 
such that \(\phi(0_\AV) = 0_{\AV'}\).
In dimension \(g = 1\),
these definitions coincide with those for elliptic curves.

The proper higher-dimensional generalization of an elliptic curve
is a \emph{principally polarized abelian variety} (PPAV).
A \emph{polarization} of \(\AV\) is an isogeny
\(\lambda : \AV \to \widehat{\AV}\),
where \(\widehat{\AV} \cong \Pic^0(\AV)\)
is the \emph{dual} abelian variety;
\(\lambda\) is \emph{principal} if it is an isomorphism.
If \(\AV = \EC\)
is an elliptic curve,
then there is a canonical principal polarization 
\(\lambda: P \mapsto [(P) - (\infty)]\),
and every other principal polarization is isomorphic to \(\lambda\)
(via composition with a suitable translation and automorphism).
The Jacobian \(\Jac{\XC}\) of a curve \(\XC\) 
also has a canonical principal polarization
defined by the theta divisor,
which essentially corresponds to an embedding of \(\XC\) in~\(\Jac{\XC}\),
and thus connects \(\Jac{\XC}\)
with the divisor class group of~\(\XC\).

%% RESTORE BELOW IN LONG VERSION
%While polarizations may seem somewhat abstract in definition,
%they correspond to projective embeddings (up to projective automorphism).
%Concretely, if we have coordinates and
%defining equations for an object, then it is polarized.

We need a notion of compatibility between isogenies and principal polarizations.
% FIXME: cite Milne (polarization and pairing sections)
First, recall that every isogeny \(\phi: \AV \to \AV'\)
has a dual isogeny \(\widehat{\phi}: \widehat{\AV'} \to \widehat{\AV}\).
Now, if \((\AV,\lambda)\) and \((\AV',\lambda')\) are PPAVs,
then \(\phi: \AV \to \AV'\) is an 
\emph{isogeny of PPAVs}
if \(\widehat{\phi}\circ\lambda'\circ\phi = [d]\lambda\)
for some integer \(d\).
We then have \(\dualof{\phi}\circ\phi = [d]\) on \(\AV\)
(and \(\phi\circ\dualof{\phi} = [d]\) on \(\AV'\)),
where \(\dualof{\phi} := \lambda^{-1}\circ\widehat{\phi}\circ\lambda'\)
is the \emph{Rosati dual}.
Intuitively,
\(\phi\) will be defined by homogeneous polynomials of degree \(d\)
with respect to projective coordinate systems
on \(\AV\) and \(\AV'\)
corresponding to \(\lambda\) and \(\lambda'\),
respectively.
There is a simple criterion on subgroups \(S \subset \AV[d]\)
to determine when an isogeny with kernel \(S\)
is an isogeny of PPAVs:
the subgroup should be \emph{Lagrangian}.\footnote{%
    Isogenies with strictly smaller kernels
    exist---isogenies with cyclic kernel
    are treated algorithmically in
    \cite{Dudeanu--Jetchev--Robert--Vuille}--- but
    these isogenies are not relevant to this investigation.
}

\begin{definition}
    \label{def:Lagrangian}
    Let \(\AV/\FFbar_p\) be a PPAV
    and let \(m\) be an integer prime to \(p\).
    A \emph{Lagrangian subgroup} of \(\AV[m]\)
    is a maximal \(m\)-Weil isotropic subgroup of \(\AV[m]\).
%    \footnote{%
%        The sudden appearance of the Weil pairing might seem surprising here,
%        but the choice of a principal polarization is crucial
%        when defining Weil pairings on an abstract abelian variety
%        from general statements about the duality between \(\AV[m]\) and \(\widehat{\AV}[m]\)
%        (see~\cite[\S16.8]{Milne} for details).
%        Thus the Weil pairing is important in determining which
%        isogenies respect principal polarizations.
%    }
\end{definition}

If \(\ell \not= p\) is prime, then
\(\AV[\ell^n] \cong (\ZZ/\ell^n\ZZ)^{2g}\)
for all \(n > 0\). 
If \(S \subset \AV[\ell]\) is Lagrangian,
then 
\(S \cong (\ZZ/\ell\ZZ)^g\).
Any Lagrangian subgroup of \(A[\ell^n]\)
is isomorphic to 
\((\ZZ/\ell\ZZ)^{n_1}\times\cdots\times(\ZZ/\ell\ZZ)^{n_g}\)
for some \(n_1 \ge \dots \ge n_g\)
with \(\sum_i n_i = gn\)
(though not every \((n_1,\ldots,n_g)\) with \(\sum_i n_i = gn\)
occurs in this way).

We now have almost everything we need to generalize supersingular
isogeny graphs from elliptic curves to higher dimension.
The elliptic curves will be replaced by PPAVs;
\(\ell\)-isogenies
will be replaced by isogenies with Lagrangian kernels in the
\(\ell\)-torsion---called
\((\ell,\ldots,\ell)\)-isogenies---and
the elliptic dual isogeny will be replaced by the Rosati dual.
It remains to define the right analogue of supersingularity 
in higher dimension, and study the resulting graphs.

%Given any two PPAVs \(\AV_1\) and \(\AV_2\),
%with principal polarizations \(\lambda_1\) and \(\lambda_2\),
%the direct product
%\(\AV_1\times\AV_2\)
%can be equipped with the principal \emph{product polarization} 
%\( 
%    \lambda_1\times\lambda_2
%    :
%    \AV_1\times\AV_2
%    \to
%    \widehat{\AV_1}\times\widehat{\AV_2}
%    \cong
%    \widehat{\AV_1\times\AV_2}
%\).
%The product polarization may be only one of many polarizations on the
%product.
%In particular, a product of supersingular elliptic curves
%can have a very high number of non-isomorphic principal polarizations,
%as we will see.

\section{%%%%%%%%%%%%%%%%%%%%%%%%%%%%%%%%%%%%%%%%%%%%%%%%%%%%%%%%%%%%%%%%%%%%%%%
    The superspecial isogeny graph in dimension \(g\)
}%%%%%%%%%%%%%%%%%%%%%%%%%%%%%%%%%%%%%%%%%%%%%%%%%%%%%%%%%%%%%%%%%%%%%%%%%%%%%%%
\label{sec:SSgraph}

%Let \(\AV/\FFbar_p\) be a PPAV of dimension \(g\).
%The \(p\)-power torsion is 
%\(\AV[p^n] \cong (\ZZ/p^n\ZZ)^r\)
%for all \(n > 0\),
%for some \(0 \le r \le g\) called the \emph{\(p\)-rank} of \(\AV\).

We need an appropriate generalization of elliptic supersingularity
to \(g > 1\).
%This is a subtle matter,
As explained in~\cite{Castryck--Decru--Smith},
it does not suffice to simply take the PPAVs \(\AV/\FFbar_{p}\) with
\(\AV[p] = 0\).

\begin{definition}
    A PPAV \(\AV\) is \textbf{supersingular}
    if the Newton polygon of its Frobenius endomorphism
    has all slopes equal to \(1/2\),
    and \textbf{superspecial}
    if Frobenius acts as \(0\) on \(H^1(\AV,\mathcal{O}_\AV)\).
    Superspecial implies supersingular;
    in dimension \(g = 1\),
    the definitions coincide.
\end{definition}

All supersingular PPAVs are isogenous to 
a product of supersingular elliptic curves.
Superspecial abelian varieties are isomorphic to 
a product of supersingular elliptic curves,
though generally only as \emph{unpolarized} abelian varieties.
The special case of Jacobians
is particularly relevant for us when constructing examples:
\(\Jac{\XC}\) is superspecial if and only if
the Hasse--Witt matrix of \(\XC\) vanishes.

It is argued in~\cite{Castryck--Decru--Smith}
that the world of superspecial (and not supersingular) PPAVs 
is the correct setting for supersingular isogeny-based cryptography.
We will not repeat this argument here;
but in any case, every higher-dimensional ``supersingular'' cryptosystem
proposed so far has in fact been superspecial.

In analogy with the elliptic supersingular graph, then,
we define
\[
    \SSset{g} 
    :=
    \left\{
        \AV : \AV/\FF_{p^2} \text{ is a superspecial \(g\)-dimensional PPAV}
    \right\}/\cong
    \,.
\]
Our first task is to estimate the size of \(\SSset{g}\).
\begin{lemma}
    \label{lemma:number-of-vertices}
    We have
    \(
        \#\SSset{g} = 
        O(p^{g(g+1)/2})
    \).
\end{lemma}
\begin{proof}
    See~\cite[\S5]{Ekedahl}.
    This follows from the Hashimoto--Ibukiyama mass formula 
    \[
        \sum_{A \in \SSset{g}} \frac{1}{\#\Aut(A)}
        =
        \prod_{i=1}^g \frac{B_{2i}}{4i}(1 + (-p)^i)
        \,,
    \]
    where \(B_{2i}\) is the \(2i\)-th Bernoulli number.
    In particular,
    \(\#\SSset{g}\) is a polynomial in \(p\)
    of degree \(\sum_{i=1}^gi = g(g+1)/2\).
    \qed
\end{proof}
Note that \(\#\SSset{g}\) grows \emph{quadratically} in
\(g\) (and exponentially in \(\log p\)):
we have \(\#\SSset{1} = O(p)\),
\(\#\SSset{2} = O(p^3)\),
\(\#\SSset{3} = O(p^6)\),
and \(\#\SSset{4} = O(p^{10})\).

For each prime \(\ell \not= p\),
we let \(\SSgraph{g}{\ell}\)
denote the (directed) graph on \(\SSset{g}\)
whose edges are \(\FFbar_p\)-isomorphism classes
of \((\ell,\cdots,\ell)\)-isogenies of PPAVs:
that is,
isogenies whose kernels are Lagrangian subgroups of the
\(\ell\)-torsion.
Superspeciality is invariant under
\((\ell,\ldots,\ell)\)-isogeny,
so to determine the degree of the vertices of \(\SSgraph{g}{\ell}\)
it suffices to enumerate the Lagrangian subgroups of a \(g\)-dimensional
PPAV.
A simple counting argument yields
Lemma~\ref{lemma:number-of-Lagrangian-subgroups}.

\begin{lemma}
    \label{lemma:number-of-Lagrangian-subgroups}
    If \(\AV/\FFbar_p\) is a \(g\)-dimensional PPAV,
    then
    the number of Lagrangian subgroups of \(\AV[\ell]\),
    and hence the number of edges leaving \(\AV\) in
    \(\SSgraph{g}{\ell}\),
    is
    \[
        N_g(\ell)
        :=
        \sum_{d=0}^g
        \qbinom{g}{d}
        \cdot
        \ell^{\binom{g-d+1}{2}}
        \,.
    \]
    (The \(\ell\)-binomial coefficient 
    \(
        \qbinom{n}{k}
        := 
        \frac{
            (n)_\ell\cdots(n-k+1)_\ell
        }{
            (k)_\ell\cdots(1)_\ell
        }
    \),
    where 
    \(
        (i)_\ell := \frac{\ell^i-1}{\ell-1}
    \),
    counts the \(k\)-dimensional subspaces of \(\FF_\ell^n\).)
    %We note that \(\qbinom{n}{k}\) is a polynomial in \(\ell\)
    %of degree \(k(n-k)\).
    In particular,
    \(\SSgraph{g}{\ell}\) is \(N_g(\ell)\)-regular;
    and
    \(N_g(\ell)\) is a polynomial in \(\ell\) of degree \(g(g+1)/2\).
\end{lemma}
%\begin{proof} : FIXME: for longer version \qed \end{proof}

We do not yet have analogues of Pizer's theorem to guarantee
that \(\SSgraph{g}{\ell}\) is Ramanujan when \(g > 1\),
though this is proven for superspecial
abelian varieties with real multiplication~\cite{Nicole}.
We therefore work on the following hypothesis:

\begin{hypothesis}
    \label{hypothesis:1}
    The graph \(\SSgraph{g}{\ell}\) is Ramanujan.
\end{hypothesis}

We need Hypothesis~\ref{hypothesis:1}
in order to obtain the following analogue of Eq.~\ref{eq:distribution-1}
(a standard random walk theorem, as in~\cite[\S3]{HLW}):
if we fix a vertex \(\AV_0\)
and consider \(n\)-step random walks \(\AV_0\to\cdots\to\AV_n\),
then
\begin{equation}
    \label{eq:distribution-g}
    \left|
        \mathrm{Pr}[\AV_n \cong \AV] - \frac{1}{\#\SSset{g}}
    \right|
    \le
    \left(
        \frac{2\sqrt{N_g(\ell)-1}}{N_g(\ell)}
    \right)^n
    \qquad
    \text{for all }
    \AV \in \SSset{g}
    \,.
\end{equation}
That is,
random walks in \(\SSgraph{g}{\ell}\)
converge exponentially quickly to the uniform
distribution: % on \(\SSset{g}\):
after \(O(\log p)\) steps in \(\SSgraph{g}{\ell}\)
we are uniformly distributed over \(\SSset{g}\).
Given specific \(\ell\) and \(g\),
we can explicitly derive
the constant hidden by the big-\(O\) 
to bound the minimum \(n\) yielding a distribution
within \(1/\#\SSset{g}\) of uniform.

\begin{remark}
    Existing proposals of higher-dimensional supersingular isogeny-based
    cryptosystems all implicitly assume (special cases of)
    Hypothesis~\ref{hypothesis:1}.
    For the purposes of attacking their underlying hard problems,
    we are comfortable making the same hypothesis.
    After all, if our algorithms are less effective because the expansion
    properties of \(\SSgraph{g}{\ell}\) are less than ideal,
    then the cryptosystems built on \(\SSgraph{g}{\ell}\) will fail to
    be effective by the same measure.
\end{remark}

%For \(g = 1\),
%we had \(\SSset{1} \subset \FF_{p^2}\)
%via the \(j\)-invariant,
%and edges can be computed using modular polynomials
%or via explicit computations with kernel subgroups on elliptic curves.
%We will return to the analogous questions
%of how to efficiently represent elements of \(\SSset{g}\),
%and how to efficiently compute isogenies,
%below.

%Similarly, elements of \(\SSset{g}\) can be realised as tuples of elements of
%\(\FF_{p^2}\) \Bnote{FIXME: really?}
%given a suitable tuple of isomorphism invariants for \(g\)-dimensional PPAVs.
%This is generally easier if we ``rigidify'' the PPAVs by adding extra structure;
%for example, if we mark the \(n\)-torsion of \(\AV\),
%then we can use theta constants of level \(n\) here.
%But if we want to work with simple isomorphism classes of PPAVs,
%then identifying \(\SSset{g}\) with a subset of \(\FF_{p^2}^n\)
%(for some \(n\)) is a more subtle matter,
%as we will see below in~\S\ref{sec:g=2}.

\section{%%%%%%%%%%%%%%%%%%%%%%%%%%%%%%%%%%%%%%%%%%%%%%%%%%%%%%%%%%%%%%%%%%%%%%%
    Superspecial cryptosystems in dimension~\(g=2\)
}%%%%%%%%%%%%%%%%%%%%%%%%%%%%%%%%%%%%%%%%%%%%%%%%%%%%%%%%%%%%%%%%%%%%%%%%%%%%%%%

Before attacking the isogeny problem in \(\SSgraph{g}{\ell}\),
we consider
some of the cryptosystems that have recently been defined in~\(\SSgraph{2}{\ell}\).
This will also illustrate some methods for computing in these graphs,
and as well as special cases of the general phenomena that can help us solve the
isogeny problem more efficiently.
For the rest of this section, therefore, we restrict to dimension~\(g=2\).

Every \(2\)-dimensional PPAV
is isomorphic (as a PPAV)
to either the Jacobian of a genus-2 curve,
or to a product of two elliptic curves.
We can therefore split \(\SSset{2}\) naturally into two disjoint subsets:
\(
    \SSset{2} = \SSset{2}^J \sqcup \SSset{2}^E
\),
where
\begin{align*}
    \SSset{2}^J
    & := 
    \left\{
        \AV \in \SSset{2}
        :
        \AV \cong \Jac{\XC}
        \text{ with }
        g(\XC) = 2
    \right\}
    \quad \text{and}
    \\
    \SSset{2}^E
    & :=
    \left\{
        \AV \in \SSset{2}
        :
        \AV \cong \EC_1 \times \EC_2
        \text{ with }
        \EC_1, \EC_2 \in \SSset{1}
    \right\}
    \,.
\end{align*}
Vertices in \(\SSset{2}^J\) are ``general'',
while vertices in \(\SSset{2}^E\) are ``special''.
We can make the estimates implied by
Lemma~\ref{lemma:number-of-vertices}
more precise:
if \(p > 5\), then
\[
    \#\SSset{2}^J
    =
    \frac{1}{2880}p^3
    + 
    \frac{1}{120}p^2
    \qquad 
    \text{and}
    \qquad
    \#\SSset{2}^E
    =
    \frac{1}{288}p^2
    +
    O(p)
\]
(see e.g.~\cite[Proposition~2]{Castryck--Decru--Smith}).
In particular,
%\(\#\SSset{2} = \frac{1}{2880}(p^3 + 34p^2) + O(p)\) and 
\(\#\SSset{2}^E/\#\SSset{2} = 10/p + o(1)\).
% alternatively, \(\#\SSset{2}^E \approx (3 - 3/25)(\#\SSset{2})^{2/3}\).

\paragraph{Takashima's hash function.}

Takashima~\cite{Takashima}
was the first to generalize CGL
to \(g = 2\).
We start with a distinguished vertex \(\AV_0\)
in \(\SSset{2}\),
and a distinguished incoming edge
\(\AV_{-1}\to \AV_0\) in \(\SSgraph{2}{\ell}\).
Each message \(m\) then drives a walk in \(\SSgraph{2}{\ell}\):
at each vertex we have a choice of 14 forward isogenies
(the 15th is the dual of the previous,
which is a prohibited backtracking step).
The message \(m\) is therefore coded in base 14.
While traversing the graph,
the vertices are handled as concrete genus-2 curves
representing the isomorphism classes of their Jacobians.
Lagrangian subgroups correspond to factorizations of the hyperelliptic
polynomials into a set of three quadratics,
and the isogenies are computed using Richelot's
formul\ae{} (see~\cite[Chapters 9-10]{Cassels--Flynn}
and \cite[Chapter 8]{Smith}).
We derive a hash value 
From the final vertex \(\AV_n\) 
as the Igusa--Clebsch invariants of the Jacobian, in \(\FF_{p^2}^3\);
Takashima does not define a finalisation map 
(into \(\FF_p^3\), for example).

Flynn and Ti observe in~\cite{flynnti}
that this hash function has a fatal weakness:
it is trivial to compute length-4 cycles starting from any vertex
in \(\SSgraph{2}{2}\),
as in Example~\ref{ex:trivial-4-cycles}.
Every cycle
produces infinitely many hash collisions.

\begin{example}
    \label{ex:trivial-4-cycles}
    Given some \(\AV_0\) in \(\SSset{2}\),
    choose a point \(P\) of order 4 on \(\AV_0\).
    There exist \(Q\) and \(R\) in \(\AV_0[2]\)
    such that 
    \(e_2([2]P,Q) = 1\)
    and \(e_2([2]P,R) = 1\),
    but \(e_2(Q,R) \not= 1\).
    The Lagrangian subgroups
    \(K_0 := \subgrp{[2]P,Q}\)
    and 
    \(K_0' := \subgrp{[2]P,R}\)
    of \(\AV_0[2]\)
    are kernels of \((2,2)\)-isogenies
    \(\phi_0: \AV_0 \to \AV_1 \cong \AV_0/K_0\)
    and 
    \(\phi_0': \AV_0 \to \AV_1' \cong \AV_0/K_0'\);
    and in general, \(\AV_1 \not\cong \AV_1'\).
    Now \(K_1 := \phi_0(K_0')\)
    and \(K_1' := \phi_0'(K_0)\)
    are Lagrangian subgroups of \(\AV_1[2]\).
    Writing \(I_1 = \ker\dualof{\phi_1}\)
    and \(I_1' = \ker\dualof{(\phi_1')}\),
    we see that \(K_1\cap I_1 = \subgrp{\phi_1(R)}\)
    and \(K_1'\cap I_1' = \subgrp{\phi_1(Q)}\).
    We thus define another pair of \((2,2)\)-isogenies,
    \(\phi_1: \AV_1 \to \AV_2 \cong \AV_1/K_1\)
    and 
    \(\phi_1': \AV_1' \to \AV_2' \cong \AV_1'/K_1'\).
    We have \(\ker(\phi_1\circ\phi_0) = \ker(\phi_1'\circ\phi_0')\),
    so \(\AV_2 \cong \AV_2'\).
    Now let \(\psi := \dualof{(\phi_0')}\circ\dualof{(\phi_1')}\circ\phi_1\circ\phi_0\).
    We have \(\psi \cong [4]_{\AV_0}\),
    but \(\psi\) does not factor over \([2]_{\AV_0}\)
    (since \(\AV_1\not\cong\AV_1'\)).
    Hence
    \(\psi\) represents a nontrivial cycle of length \(4\) in the graph.
\end{example}

The ubiquity of these length-4 cycles does not mean that \(\SSgraph{2}{2}\) is no use for hashing:
it just means that we must use a stronger rule than 
backtrack-avoidance when selecting steps in a walk.
The following hash function does just this.

\paragraph{The Castryck--Decru--Smith hash function (CDS).}

Another generalization of CGL 
from \(\SSgraph{1}{2}\) to \(\SSgraph{2}{2}\),
neatly avoiding the length-4 cycles of
Example~\ref{ex:trivial-4-cycles},
is defined in~\cite{Castryck--Decru--Smith}.
Again, we fix a vertex \(\AV_0\)
and an isogeny \(\phi_{-1}: \AV_{-1} \to \AV_0\);
we let \(I_0 \subset \AV_0[2]\) be 
the kernel of the Rosati dual \(\dualof{\phi}_{-1}\).
Now, let \(m = (m_0,\ldots,m_{n-1})\)
be a \(3n\)-bit message,
with each \(0 \le m_i < 8\).
The sequence \((m_0,\ldots,m_{n-1})\)
drives a path through \(\SSgraph{2}{2}\)
as follows:
our starting point is \(\AV_0\),
with its distinguished subgroup \(I_0\)
corresponding to the edge \(\AV_{-1} \to \AV_0\).
For each \(0 \le i < n\),
we compute the set of eight Lagrangian subgroups \(\{S_{i,0},\ldots,S_{i,7}\}\) of \(\AV_i[2]\)
such that \(S_{i,j} \cap I_i = 0\),
numbering them according to some fixed ordering on the
encodings of Lagrangian subgroups.
Then we compute 
\(\phi_i: \AV_i \to \AV_{i+1} \cong \AV_i/S_{i,m_i}\),
and let \(I_{i+1} := \phi_i(\AV_i[2]) = \ker\dualof{\phi_i}\).
Once we have computed the entire walk \(\AV_0 \to \cdots \to \AV_n\),
we can derive a \(3\log_2p\)-bit hash value \(H(m)\) from
the isomorphism class of \(\AV_n\)
(though such a finalisation is unspecified in~\cite{Castryck--Decru--Smith}).
The subgroup intersection condition ensures that the 
composition of the isogenies in the walk is a
\((2^n,\ldots,2^n)\)-isogeny,
thus protecting us from the small cycles of
Example~\ref{ex:trivial-4-cycles}.

Putting this into practice
reveals an ugly technicality.
As in Takashima's hash function,
we compute with vertices as genus-2 curves,
encoded by their hyperelliptic polynomials,
with 
%kernel subgroups encoded as quadratic splittings and 
\((2,2)\)-isogenies computed using Richelot's formul\ae{}.
Walk endpoints are mapped to 
Igusa--Clebsch invariants in \(\FF_{p^2}^3\).
But these curves, formul\ae{}, and invariants
only exist for vertices in \(\SSset{2}^J\).
We can handle vertices in \(\SSset{2}^E\)
as pairs of elliptic curves,
with pairs of \(j\)-invariants for endpoints,
and there are explicit formul\ae{} to compute
isogenies in to and out of \(\SSset{2}^E\)
(see e.g.~\cite[\S3]{Castryck--Decru--Smith}).
Switching between representations and algorithms
(to say nothing of finalisation,
where \(\SSset{2}^E\) would have a smaller, easily distinguishable,
and easier-to-invert image)
seems like needless fiddle when
the probability of stepping onto a vertex in \(\SSset{2}^E\)
is only \(O(1/p)\),
which is negligible for cryptographic~\(p\).

In~\cite{Castryck--Decru--Smith}, this issue
was swept under the rug
by defining simpler algorithms which efficiently walk in
the subgraph of \(\SSgraph{2}{2}\) supported on \(\SSset{2}^J\),
and simply fail if they walk into \(\SSset{2}^E\).
This happens with probability \(O(1/p)\),
which may seem acceptable---however, 
this also means that
\emph{it is exponentially easier to find a message where the hash fails
than it is to find a preimage} with a square-root algorithm.
The former requires \(O(p)\) work,
the latter \(O(p^{3/2})\).
In this, as we will see,
the simplified CDS hash function
contains the seeds of its own destruction.

\paragraph{Genus-2 SIDH.}
%In~\cite{costellokummer},
%the first author used isogenies on abelian surfaces---specifically,
%Weil restrictions of supersingular elliptic curves---with the aim of accelerating
%parts of SIDH computations.
%In essence, one can transform a walk in \(\SSgraph{1}{2}\)
%into a walk in a special subgraph of \(\SSgraph{2}{2}\).

Flynn and Ti~\cite{flynnti} defined an SIDH analogue in dimension \(g = 2\).
As in the hash functions above,
Richelot isogenies are used for Alice's steps in \(\SSgraph{2}{2}\),
while explicit formul\ae{} for \((3,3)\)-isogenies on Kummer surfaces
are used for Bob's steps in \(\SSgraph{2}{3}\).
Walks may (improbably) run into \(\SSset{2}^E\),
as with the hash functions above;
but the same work-arounds apply without affecting security.
(Further, if we generate a public key in \(\SSset{2}^E\),
then we can discard it and generate a new one in \(\SSset{2}^J\).)
%(Further, if we hit \(\SSset{2}^E\) during key generation,
%then we are free to fail and re-start the algorithm.)
As with SIDH, breaking public keys 
amounts to computing \emph{short} solutions to the isogeny problem
in \(\SSgraph{2}{2}\) or \(\SSgraph{2}{3}\),
though presumably endomorphism attacks generalizing~\cite{petit}
also exist.

%\begin{remark}
%    Each of the cryptosystems above implicitly works in the subgraph of
%    \(\SSgraph{2}{\ell}\) supported on \(\SSset{2}^J\),
%    the vertices corresponding to Jacobians of genus-2 curves,
%    either for algorithmic simplicity (as in Takashima and Flynn--Ti),
%    or by definition (as in CDS).
%    Working exclusively with vertices representable by Jacobians
%    is possible for \(g = 2\) because, as noted above,
%    a general \(2\)-dimensional PPAV is isomorphic to a Jacobian.
%    In higher dimensions, we are not so lucky.
%    For \(g = 3\),
%    a general PPAV is isomorphic to a Jacobian~\cite{Oort--Ueno},
%    but possibly over an inconvenient quadratic extension.
%    For \(g > 3\), however,
%    a general PPAV looks nothing like a Jacobian.
%    Indeed, 
%    the moduli of Jacobians form a \((3g-3)\)-dimensional
%    subspace of the \(g(g+1)/2\)-dimensional space of PPAVs:
%    the Jacobians go from being universal in universal in \(g = 1\)
%    and ubiquitous in \(g = 2\)
%    to vanishingly rare in \(g > 3\).
%    Indeed, there is no particular reason why a quotient of a
%    \(g\)-dimensional Jacobian by a Lagrangian subgroup 
%    should be another Jacobian when \(g > 3\);
%    the subgraph of \(\SSgraph{g}{\ell}\)
%    supported on \(\SSset{g}^J\) is therefore not only vanishingly
%    small, but it may well have fewer edges than vertices.
%\end{remark}

\section{%%%%%%%%%%%%%%%%%%%%%%%%%%%%%%%%%%%%%%%%%%%%%%%%%%%%%%%%%%%%%%%%%%%%%%%
    Attacking the isogeny problem in superspecial graphs
}%%%%%%%%%%%%%%%%%%%%%%%%%%%%%%%%%%%%%%%%%%%%%%%%%%%%%%%%%%%%%%%%%%%%%%%%%%%%%%%

We want to solve the isogeny problem in \(\SSgraph{g}{\ell}\).
We can always do this using random walks
in \(O(\sqrt{\#\SSset{g}}) = O(p^{g(g+1)/4})\) classical steps.
%\Bnote{FIXME: quantum cost. CC: Grovering the full set directly would also be square root set size, so no improvement. This is analogous to Biasse--Jao--Sankar citing $p^{1/2}$ complexity for both classical and quantum before their paper. Fortunately, it's our style of algorithm (rather than solving path finding directly) where you're searching for a set in a much bigger set where Grover is powerful. Summary -- we need not say anything quantum here.}

Our idea is that
\(\SSset{g-1}\times\SSset{1}\)
maps into \(\SSset{g}\)
by mapping a pair of PPAVs
to their product equipped with the product polarization,
and the image of \(\SSset{g-1}\times \SSset{1}\)
represents a large set of easily-identifiable ``distinguished vertices''
in \(\SSgraph{g}{\ell}\).
Indeed,
since the map 
\(\SSset{g-1}\times\SSset{1}\to\SSset{g}\)
is generically finite, of degree independent of \(p\),
Lemma~\ref{lemma:number-of-vertices}
implies that
\begin{equation}
    \label{eq:size-relation}
    \#\SSset{g}/\#(\text{image of } \SSset{g-1}\times\SSset{1}) = O(p^{g-1})
    \qquad
    \text{for } g > 1
    \,.
\end{equation}
We can efficiently detect such a step into a product PPAV
in a manner analogous to that of the failure of the CDS hash function:
for example, by the breakdown of 
a higher-dimensional analogue of Richelot's formul\ae{} such
as~\cite{Lubicz--Robert}.

We can walk into this subset,
then recursively solve the path-finding problem in
the subgraphs
\(\SSgraph{g-1}{\ell},\ldots,\SSgraph{1}{\ell}\)
(each time walking from \(\SSgraph{i}{\ell}\)
into \(\SSgraph{i-1}{\ell}\times\SSgraph{1}{\ell}\))
before gluing the results together
to obtain a path in \(\SSgraph{g}{\ell}\).

\begin{lemma}
    \label{lemma:combine}
%    Let \(\alpha: A \to A'\)
%    be a walk of length \(a\) in \(\SSgraph{i}{\ell}\)
%    and \(\beta: B \to B'\)
%    be a walk of length \(b\) in \(\SSgraph{j}{\ell}\)
%    for some \(a\) and \(b\).
%    Without loss of generality,
%    suppose \(a \ge b\).
    Let \(\alpha: A \to A'\)
    and \(\beta: B \to B'\)
    be walks 
    in \(\SSgraph{i}{\ell}\)
    and \(\SSgraph{j}{\ell}\)
    of lengths \(a\) and \(b\),
    respectively.
    If \(a \equiv b \pmod{2}\),
    then 
    we can efficiently compute
    a path of length \(\max(a,b)\) from
    \(A\times B\) to \(A'\times B'\)
    in \(\SSgraph{i+j}{\ell}\).
\end{lemma}
\begin{proof}
    Write 
    \(\alpha = \alpha_1\circ\cdots\circ\alpha_a\)
    and
    \(\beta = \beta_1\circ\cdots\circ\beta_b\)
    as compositions of \((\ell,\cdots,\ell)\)-isogenies.
    %% RESTORE "Without..." in place of WLOG
    %Without loss of generality,
    WLOG,
    suppose \(a \ge b\).
    Set \(\beta_{b+1} = \dualof{\beta_b}\),
    \(\beta_{b+2} = \beta_b\),
    ...,
    \(\beta_{a-1} = \dualof{\beta_b}\), 
    \(\beta_{a} = \beta_b\);
    then
    \(\alpha\times\beta:
    (\alpha_1\times\beta_1)\circ\cdots\circ(\alpha_a\times\beta_a)\)
    is a path from \(A\times B\) to \(A'\times B'\).
    \qed
\end{proof}

Equations~\ref{eq:distribution-g} and~\ref{eq:size-relation}
show that a walk of length \(O(\log p)\)
lands in the image of \(\SSset{g-1}\times\SSset{1}\)
with probability \(O(1/p^{g-1})\),
and after \(O(p^{g-1})\) such short walks
we are in \(\SSset{g-1}\times\SSset{1}\)
with probability bounded away from zero.
More generally,
we can walk into the image of \(\SSset{g-i}\times\SSset{i}\)
for any \(0 < i < g\);
but
the probability of this is \(O(1/p^{i(g-i)})\),
which is maximised by \(i = 1\) and \(g-1\).

\begin{algorithm}
    \caption{Computing isogeny paths in \(\SSgraph{g}{\ell}\)}
    \label{alg:general}
    \KwIn{\(\AV\) and \(\AV'\) in \(\SSset{g}\)}
    \KwOut{A path \(\phi: \AV \to \AV'\) in \(\SSgraph{g}{\ell}\)}
    %\BlankLine
    %
    Find a path \(\psi\)
    from \(\AV\) to some point \(\BV\times \EC\)
    in \(\SSset{g-1}\times \SSset{1}\)     \label{alg:general:step-1}
    \;
    Find a path \(\psi'\)
    from \(\AV'\) to some point \(\BV'\times \EC'\)
    in \(\SSset{g-1}\times \SSset{1}\)    \label{alg:general:step-2}
    \;
    Find a path \(\beta: \BV \to \BV'\) in \(\SSgraph{g-1}{\ell}\)    \label{alg:general:step-3}
    using Algorithm~\ref{alg:general} recursively if \(g-1 > 1\),
    or elliptic path-finding if \(g-1 = 1\)
    \;
    Find a path \(\eta: \EC \to \EC'\) in \(\SSgraph{1}{\ell}\)       \label{alg:general:step-4}
    using elliptic path-finding
    \;
    Let \(b = \text{length}(\beta)\)
    and \(e = \text{length}(\eta)\).
    If \(b \not\equiv e \pmod{2}\),
    then fail and return \(\bot\)
    (or try again with another \(\psi\)
    and/or \(\psi'\), \(\beta\), or \(\eta\))                        \label{alg:general:step-5}
    \;
    Construct the product path \(\pi: \BV\times \EC\to \BV'\times \EC'\)
    defined by Lemma~\ref{lemma:combine}.
    \;
    \Return the path                                               \label{alg:general:step-6}
        \(
            \phi := \dualof{\psi'} \circ \pi \circ \psi
        \)
        from \(\AV\) to \(\AV'\).
\end{algorithm}

\subsubsection{Proof of Theorem~\ref{thm:classic}}
Algorithm~\ref{alg:general}
implements the approach above, and proves Theorem~\ref{thm:classic}.
Step~\ref{alg:general:step-1} computes \(\psi\)
by taking \(O(p^{g-1})\)
non-backtracking random walks of length \(O(\log(p))\)
which can be trivially parallelized,
so with \(P\) processors we expect \(\softO(p^{g-1}/P)\) steps
before finding \(\psi\).
(If \(\AV\) is a fixed public base point %in \(\SSset{g}\)
then we can assume \(\psi\) is already known).
Likewise, Step~\ref{alg:general:step-2} takes
\(\softO(p^{g-1}/P)\) steps to compute \(\psi'\).
After \(g-1\) recursive calls,
we have reduced to the problem of computing paths in
\(\SSgraph{1}{\ell}\)
in Step~\ref{alg:general:step-4},
which can be done in time \(O(\sqrt{p}/P)\).
Step~\ref{alg:general:step-6}
applies Lemma~\ref{lemma:combine}
to compute the final path in polynomial time.
At each level of the recursion,
we have a \(1/2\) chance of having the same walk-length parity;
hence, Algorithm~\ref{alg:general} succeeds with probability \(1/2^{g-1}\).
This could be improved by computing more walks when
the parities do not match, but \(1/2^{g-1}\) suffices to prove the
theorem.
The total runtime is
\(
    %\softO\big(
        %\frac{1}{P}
        %\big(\sum_{i=1}^{g} p^{i-1} + \sqrt{p} \big)
    %\big)
    %= 
    \softO(p^{g-1}/P)
\)
isogeny steps.

\subsubsection{Proof of Theorem~\ref{thm:quantum}}

Algorithm~\ref{alg:general} can be run in a quantum computation
model as follows.
First,
recall from the proof of Theorem~\ref{thm:classic} that
Steps~\ref{alg:general:step-1} and~\ref{alg:general:step-2} 
find product varieties by taking $O(p^{g-1})$ walks of length $O(\log(p))$.
Here we proceed following Biasse, Jao and Sankar~\cite[\S 4]{BJS14}.
Let $N$ be the number of walks in $O(p^{g-1})$ of length $\lambda$ (in $O(\log(p))$). 
To compute $\psi$, we define an injection
\[
    f \colon [1, \dots , N] \longrightarrow 
    \{ \text{nodes of distance $\lambda$ starting from $\AV$} \}
    \,,
\]
and a function $C_f\colon [1, \dots , N] \to \{0,1\}$ by
\( C_f(x) = 1\) if \(f(x)\) is in \(\SSset{g-1} \times \SSset{1}\),
and \(0\) otherwise.
If there is precisely one $x$ with $C_f(x)=1$,
Grover's algorithm~\cite{grover} will find it (with probability \(\ge 1/2\))
in $O(\sqrt{N})$ iterations.
If there are an unknown $t \ge 1$ such solutions,
then Boyer--Brassard--H{\o}yer--Tapp~\cite{boyer1998tight} finds one in
$O(\sqrt{N/t})$ iterations. Hence, if we take $\lambda$ large enough to
expect at least one solution, then we will find it in
$O(\sqrt{p^{g-1}})$ Grover iterations. We compute $\psi'$ (and any
recursive invocations of Steps~\ref{alg:general:step-1}
and~\ref{alg:general:step-2}) similarly.

For the elliptic path finding in Steps~\ref{alg:general:step-3}
and~\ref{alg:general:step-4}, we can apply (classical) Pollard-style
pseudorandom walks which require $\softO(\sqrt{p})$ memory and
$\softO(\sqrt{p})$ operations to find an $\ell$-isogeny path.
Alternatively, we can reduce storage costs by applying Grover's
algorithm to the full graph \(\SSgraph{1}{\ell}\) to find an $\ell$-isogeny path in expected time $O(\sqrt{p})$. 
Finally, Step~\ref{alg:general:step-6} applies Lemma~\ref{lemma:combine} to compute the final path.

\begin{remark}
    \label{rem:endomorphism-algorithm}
    We can use the same approach as Algorithm~\ref{alg:general}
    to compute explicit endomorphism rings of superspecial PPAVs.
    Suppose we want to compute \(\End(\AV)\) for some \(g\)-dimensional
    \(\AV\) in \(\SSset{g}\).
    Following the first steps of Algorithm~\ref{alg:general},
    we compute a walk \(\phi\) from \(\AV\) into \(\SSset{g-1}\times\SSset{1}\),
    classically or quantumly,
    recursing until we end up at some \(\EC_1\times\cdots\times\EC_g\) in \(\SSset{1}^g\).
    Now we apply an elliptic endomorphism-ring-computing algorithm
    to each of the \(\EC_i\);
    this is equivalent to solving the isogeny problem
    in~\(\SSgraph{1}{\ell}\)
    (see~\cite[\S5]{petit}),
    so its cost is in \(\softO(\sqrt{p})\).
    The products of the generators for the \(\End(\EC_i)\) 
    form generators for
    \(\End(\EC_1\times\cdots\times\EC_g)\),
    which we can then pull back through \(\phi\) 
    to compute a finite-index subring of \(\End(\AV)\)
    that is maximal away from \(\ell\).
    The total cost is a classical \(\softO(p^{g-1}/P)\)
    (on \(P\) processors),
    or a quantum \(\softO(\sqrt{p^{g-1}})\),
    plus the cost of the pullback.
\end{remark}

\begin{remark}\label{rem:anyversusell}
    Algorithm~\ref{alg:general}
    computes compositions of $(\ell, \dots, \ell)$-isogenies.
    If we relax and allow arbitrary-degree isogenies,
    not just paths in \(\SSgraph{g}{\ell}\) for fixed $\ell$, 
    then the elliptic path-finding steps can use
    the classical Delfs--Galbraith~\cite{Delfs--Galbraith} 
    or quantum Biasse--Jao--Sankar~\cite{BJS14} algorithms.
    While this would not change the asymptotic runtime of Algorithm~\ref{alg:general}
    (under the reasonable assumption 
    that the appropriate analogue of vertices ``defined over \(\FF_p\)'' with
    commutative endomorphism rings form 
    a subset of size \(O(\sqrt{\#\SSset{g}})\)),
    both of these algorithms have low memory requirements and are arguably more implementation-friendly than Pollard-style pseudorandom walks~\cite[\S 4]{Delfs--Galbraith}.
\end{remark}

\section{%%%%%%%%%%%%%%%%%%%%%%%%%%%%%%%%%%%%%%%%%%%%%%%%%%%%%%%%%%%%%%%%%%%%%%%
    Cryptographic implications
}%%%%%%%%%%%%%%%%%%%%%%%%%%%%%%%%%%%%%%%%%%%%%%%%%%%%%%%%%%%%%%%%%%%%%%%%%%%%%%%

Table~\ref{tab:complexities} compares 
Algorithm~\ref{alg:general} with the best known attacks 
for dimensions $g\le 6$.
For general path-finding,
the best known algorithms are
classical Pollard-style pseudorandom walks
and quantum Grover search~\cite{grover,boyer1998tight}.
As noted in Remark~\ref{rem:anyversusell},
higher-dimensional analogues of
Delfs--Galbraith~\cite{Delfs--Galbraith}
or Biasse--Jao--Sankar~\cite{BJS14}
might yield practical improvements,
without changing the asymptotic runtime.

\begin{table}
    \caption{Logarithms (base \(p\))
    of asymptotic complexities of %classical and quantum 
    algorithms for solving the isogeny problems in \(\SSgraph{g}{\ell}\)
    for \(1 \le g \le 6\). Further explanation in text.}
    \label{tab:complexities}
    \renewcommand{\tabcolsep}{0.20cm}
    \renewcommand{\arraystretch}{1.3}
    \centering
    \begin{tabular}{crcccccc}
        & Dimension \(g\) & 1 & 2 & 3 & 4 & 5 & 6 \\
        \hline
        \multirow{2}{*}{Classical} &
        {\bf Algorithm~\ref{alg:general}} &
        --- & \(\mathbf{1}\) & \(\mathbf{2}\) & \(\mathbf{3}\) & \(\mathbf{4}\) & \(\mathbf{5}\)
        \\
        &
        Pollard/Delfs--Galbraith~\cite{Delfs--Galbraith} &
        \(0.5\) & 
        \(1.5\) & 
        \(3\) & 
        \(5\) & 
        \(7.5\) & 
        \(10.5\) 
        \\
        \hline
        \multirow{2}{*}{Quantum} &
        {\bf Algorithm~\ref{alg:general}} &
        --- & \(\mathbf{0.5}\) & \(\mathbf{1}\) & \(\mathbf{1.5}\) & \(\mathbf{2}\) & \(\mathbf{2.5}\)
        \\
        &
        Grover/Biasse--Jao--Sankar~\cite{BJS14} &
        \(0.25\) &
        \(0.75\) &
        \(1.5\) &
        \(2.5\) &
        \(3.75\) &
        \(4.25\)
        \\
        \bottomrule 
    \end{tabular}
\end{table}

The paths in \(\SSgraph{g}{\ell}\)
constructed by Algorithm~\ref{alg:general} 
are generally too long to be private keys for SIDH analogues,
which are paths of a fixed and typically shorter length.
Extrapolating from 
$g=1$~\cite{SIDH} and $g=2$~\cite{flynnti},
we suppose that the secret keyspace has size $O(\sqrt{\#\SSset{g}}) = O(p^{g(g+1)/4})$
and the target isogeny has degree in $O(\sqrt{p})$,
corresponding to a path of length roughly \(\log_\ell(p)/2\) in \(\SSgraph{g}{\ell}\).
On the surface, therefore,
Algorithm~\ref{alg:general} does not yield a direct attack on
SIDH-style protocols;
or, at least, not a direct attack that succeeds with high probability.
(Indeed, to resist direct attacks from Algorithm~\ref{alg:general},
it would suffice to abort any key generations passing through
vertices in \(\SSset{g-1}\times\SSset{1}\).)

However, we can anticipate an attack via endomorphism rings, 
generalizing the attack described at the end of~\S\ref{sec:crypto-g-1},
using the algorithm outlined in Remark~\ref{rem:endomorphism-algorithm}.
If we assume that what is polynomial-time for elliptic endomorphisms
remains so for (fixed) \(g > 1\),
then we can break \(g\)-dimensional SIDH keys
by computing shortest paths in~\(\SSgraph{g}{\ell}\)
with the same complexity as Algorithm~\ref{alg:general}:
that is, classical \(\softO(p^{g-1}/P)\)
and quantum \(\softO(p^{(g-1)/2})\) 
for \(g > 1\).

This conjectural cost compares very favourably
against the best known classical and quantum attacks
on \(g\)-dimensional SIDH.
In the classical paradigm,
a meet-in-the-middle attack
would run in \(\softO(p^{g(g+1)/8})\),
with similar storage requirements.
In practice the best attack 
is the golden-collision van Oorschot--Wiener (vOW) algorithm~\cite{vOW}
investigated in~\cite{Adjetal},
which given storage \(w\)
runs in expected time \(\softO(p^{3g(g+1)/16}/(P\sqrt{w}))\). For fixed $w$, the attack envisioned above gives an asymptotic improvement over vOW for all $g>1$. If an adversary has access to a large amount of storage, then vOW may still be the best classical algorithm for $g \leq 5$, particularly when smaller primes are used to target lower security levels. (vOW becomes strictly worse for all $g>5$, even if we assume unbounded storage.) In the quantum paradigm,
%Grover search requires \(\softO(p^{g(g+1)/8})\) steps;
Tani's algorithm~\cite{tani}
would succeed in \(\softO(p^{g(g+1)/12})\), meaning we get the same asymptotic complexities for dimensions 2 and 3, and an asymptotic improvement for all $g>3$. Moreover, Jaques and Schanck~\cite{JS19} suggest a significant gap between the asymptotic runtime of Tani's algorithm and its actual efficacy in any meaningful model of quantum computation. On the other hand, the bottleneck of the quantum attack forecasted above is a relatively straightforward invocation of Grover search, and the gap between its asymptotic and concrete complexities is likely to be much closer. 

Like the size of $S_g(p)$, the exponents in the runtime complexities of
all of the algorithms above are quadratic in $g$. 
Indeed, this was the practical motivation for instantiating isogeny-based
cryptosystems in \(g > 1\).
In contrast, the exponents for
Algorithm~\ref{alg:general} and our proposed SIDH attack
are linear in $g$. 
This makes the potential trade-offs for cryptosystems based on
higher-dimensional supersingular isogeny problems
appear significantly less favourable, particularly as \(g\) grows 
and the gap between the previous best attacks and Algorithm~\ref{alg:general} widens.

%%%%%%%%%%%%%%%%%%%%%%%%%%%%%%%%%%%%%%%%%%%%%%%%%%%%%%%%%%%%%%%%%%%%%%%%%%%%%%%%
%\bibliographystyle{splncs04}
\bibliographystyle{abbrv}
\bibliography{bib}
%%%%%%%%%%%%%%%%%%%%%%%%%%%%%%%%%%%%%%%%%%%%%%%%%%%%%%%%%%%%%%%%%%%%%%%%%%%%%%%%

\appendix

\section{A proof-of-concept implementation}

We include a naive Magma implementation of the product finding stage (i.e. Steps~\ref{alg:general:step-1}-\ref{alg:general:step-3}) of Algorithm~\ref{alg:general} in dimension $g=2$ with $\ell=2$.
%For a fixed $p$, 
First, it generates a challenge by walking from the known superspecial node corresponding to
the curve $\XC \colon y^2=x^5+x$ over a given $\FF_{p^2}$ to a random
abelian surface in \(\SSgraph{2}{2}\), which becomes the target \(\AV\).
Then it starts computing random walks of length slightly larger than
$\log_2(p)$, whose steps correspond to $(2,2)$-isogenies.
As each step is taken, it checks whether we have landed on a product of
two elliptic curves (at which point it will terminate) before continuing. 

Magma's built-in functionality for $(2,2)$-isogenies makes this rather straightforward. At a given node, the function {\tt RichelotIsogenousSurfaces} computes all 15 of its neighbours, so our random walks are simply a matter of generating enough entropy to choose one of these neighbours at each of the $O(\log(p))$ steps. For the sake of replicability, we have used Magma's inbuilt implementation of SHA-1 to produce pseudo-random walks that are deterministically generated by an input seed. SHA-1 produces 160-bit strings, which correspond to 40 integers in $[0,1,\dots, 15]$; this gives a straightforward way to take 40 pseudo-random steps in \(\SSgraph{2}{2}\), where no step is taken if the integer is 0, and otherwise the index is used to choose one of the 15 neighbours. 

The seed {\tt processor} can be used to generate independent walks across multiple processors. We always used the seed ``0'' to generate the target surface, and set {\tt processor} to be the string ``1'' to kickstart a single process for very small primes. For the second and third largest primes, we used the strings ``1'', ``2'', \dots , ``16'' as seeds to 16 different deterministic processes. For the largest prime, we seeded 128 different processes. 

For the prime $p={\bf 127}=2^7-1$, the seed ``0'' walks us to the starting node corresponding to $C_0/\FF_{p^2} \colon y^2=(41i + 63)x^6 +\dots + (6i +12)x + 70$. The single processor seeded with ``1'' found a product variety $E_1 \times E_2$ on its second walk after taking {\bf 53 steps} in total, with $E_1/\FF_{p^2} \colon y^2 = x^3 + (93i + 43)x^2 + (23i + 93)x + (2i + 31)$ and $E_2/\FF_{p^2} \colon y^2 = x^3 + (98i + 73)x^2 + (30i + 61)x + (41i + 8)$. 

For the prime $p={\bf 8191}=2^{13}-1$, the single processor seeded with ``1''  found a product variety on its 175-th walk after taking {\bf 6554 steps} in total. 

For the prime $p={\bf 524287}=2^{19}-1$, all 16 processors were used. The processor seeded with ``2'' was the first to find a product variety on its 311-th walk after taking 11680 steps in total. Given that all processors walk at roughly the same pace, at this stage we would have walked close to $16 \cdot 11680 = {\bf 186880}$ {\bf steps}.

For the 25-bit prime $p={\bf 17915903}=2^{13}3^7-1$, the processor seeded with ``13'' found a product variety after taking 341 walks and a total of 12698 steps. At this stage the 16 processors would have collectively taken around {\bf 203168 steps}. 

The largest experiment that we have conducted to date is with the prime $p={\bf 2147483647}=2^{31}-1$, where 128 processors walked in parallel. Here the processor seeded with ``95'' found a product variety after taking 10025 walks and a total of 375703 steps. At this stage the processors would have collectively taken around {\bf 48089984 steps}. 

In all of the above cases we see that product varieties are found with around $p$ steps. The Magma script that follows can be used to verify the experiments\footnote{Readers without access to Magma can make use of the free online calculator at \url{http://magma.maths.usyd.edu.au/calc/}, omitting the ``Write'' functions at the end that are used to print to local files.}, or to experiment with other primes. 

\newpage

\begin{tiny}

\begin{Verbatim}

//////////////////////////////////////////////////////////
clear;

processor:="1";

p:=2^13-1;
Fp:=GF(p);
Fp2<i>:=ExtensionField<Fp,x|x^2+1>;
_<x>:=PolynomialRing(Fp2);

//////////////////////////////////////////////////////////

Next_Walk := function(str)
    H := SHA1(str);
	steps := [ StringToInteger(x, 16): x in ElementToSequence(H) | x ne "0"];
    return steps ,H;
end function;

//////////////////////////////////////////////////////////

Walk_To_Starting_Jacobian:=function(str)
	
	steps,H:= Next_Walk(str);
	
	C0:=HyperellipticCurve(x^5+x);
	J0:=Jacobian(C0);
	for i:=1 to #steps do
		neighbours:=RichelotIsogenousSurfaces(J0);
		if Type(neighbours[steps[i]]) ne SetCart then
			J0:=neighbours[steps[i]];
		end if;
	end for;

	return J0;

end function;

//////////////////////////////////////////////////////////

Walk_Until_Found:=function(seed,J0);

	found:=false;
	H:=seed;
	found:=false;
	walks_done:=0;
	steps_done:=0;

	while not found do

		walks_done+:=1;
		walks_done, "walks and",steps_done, "steps on core", processor, "for p=",p;
		J:=J0;
		steps,H:=Next_Walk(H);

		for i:=1 to #steps do
			steps_done+:=1;
			J:=RichelotIsogenousSurfaces(J)[steps[i]];
			if Type(J) eq SetCart then
				found:=true;
				index:=i;
				break;
			end if;
		end for;

	end while;

	return steps,index,walks_done,steps_done,J;

end function;

//////////////////////////////////////////////////////////

file_name:="p" cat IntegerToString(p)  cat "-" cat processor cat ".txt";
J0:=Walk_To_Starting_Jacobian("0");
steps,index,walks_done,steps_done,J:=Walk_Until_Found(processor,J0);

Write(file_name, "walks done =");
Write(file_name, walks_done);
Write(file_name, "steps_done =");
Write(file_name, steps_done);
Write(file_name, "steps=");
Write(file_name, steps);
Write(file_name, "index=");
Write(file_name, index);
Write(file_name, "Elliptic Product=");
Write(file_name, J);

//////////////////////////////////////////////////////////

\end{Verbatim}

\end{tiny}

\end{document}